\documentclass[10pt,conference]{IEEEtran}
\usepackage{bm}
\usepackage{mathrsfs}
\usepackage{textcomp}
\usepackage{epsfig}
\usepackage{graphicx}
\usepackage{epstopdf}
\usepackage{indentfirst}
\usepackage{amsmath}
\usepackage{amssymb}
\usepackage{array}
\usepackage{multirow}
\usepackage{graphicx}
\usepackage{subfigure}
\usepackage{color}

\usepackage{amsthm}
\newtheorem{theorem}{Theorem}
\newtheorem{lemma}[theorem]{Lemma}
\usepackage{times}
\usepackage[ruled,vlined]{algorithm2e}
\theoremstyle{definition}

\theoremstyle{remark}

\usepackage{cite}

\begin{document}

%
\title{
Optimizing Non-Orthogonal Multiple Access in Random Access Networks
}

\author{
   \IEEEauthorblockN{Ziru Chen\IEEEauthorrefmark{1}, Yong Liu\IEEEauthorrefmark{2}, Sami Khairy \IEEEauthorrefmark{1}, Lin X. Cai\IEEEauthorrefmark{1}, Yu Cheng\IEEEauthorrefmark{1}, and Ran Zhang\IEEEauthorrefmark{3}}
    \IEEEauthorblockA{\IEEEauthorrefmark{1}Department of Electrical and Computer Engineering, Illinois Institute of Technology, Chicago, USA }
    \IEEEauthorblockA{\IEEEauthorrefmark{2} Department of Physics and Telecommunication Engineering, South China Normal University, Guangzhou, China}   
    \IEEEauthorblockA{\IEEEauthorrefmark{3}College of Engineering and Computing, Miami University, Oxford, OH, USA}
\IEEEauthorblockA{zchen71@hawk.iit.edu, yliu@m.scnu.edu.cn, skhairy@hawk.iit.edu, \{lincai, cheng\}@iit.edu, zhangr43@miamioh.edu}
}

\maketitle

\begin{abstract}
Non-orthogonal multiple access (NOMA) has been considered as a promising solution for improving the spectrum efficiency of next-generation wireless networks. In this paper, the performance of a p-persistent slotted ALOHA system in support of NOMA transmissions is investigated. Specifically, wireless users can choose to use high or low power for data transmissions with certain probabilities. To achieve the maximum network throughput, an analytical framework is developed to analyze the successful transmission probability of NOMA and long term average throughput of users involved in the non-orthogonal transmissions. The feasible region of the maximum number of concurrent users using high and low power to ensure successful NOMA transmissions are quantified. Based on analysis, an algorithm is proposed to find the optimal transmission probabilities for users to choose high and low power to achieve the maximum system throughput. In addition, the impact of power settings on the network performance is further investigated. Simulations are conducted to validate the analysis.
\end{abstract}


\section{Introduction}


 With the explosive growth of the connected devices,  machine type communications (MTCs) attract great attention of networking research community and become one of the most critical applications in 5G and beyond~\cite{dai2015non, bockelmann2016massive}. Non-orthogonal multiple access (NOMA) technology, which can effectively improve the spectral efficiency by exploiting successive interference cancellation (SIC) technique to enable non-orthogonal data transmissions, has been proved to be a good candidate in MTC to support the ultra-dense connectivity. 

 NOMA technique has been extensively studied in the literature~\cite{ding2017application,ding2017survey, maraqa2019survey,ni2019analysis}. In~\cite{ding2015impact,al2017optimum}, NOMA is applied in a centralized network, where a scheduler carefully pairs users and controls their transmission power based on their channel conditions to maximize the sum rate of NOMA users. However, for 5G MTC, it may be too costly for low power machines to measure and update their channel conditions with the scheduler, which makes it challenging to apply NOMA for MTC. Some recent works propose to apply NOMA in a slotted ALOHA system to support MTC of IoT devices~\cite{choi2017noma,seo2018nonorthogonal,choi2018game}, where users are randomly paired when they distributively access the channel. Simulation results in~\cite{elkourdi2018enabling} show that the performance of slotted ALOHA can be greatly improved when NOMA technique is applied. A game-theoretic approach is proposed in~\cite{choi2018game} to determine the transmission probability for a NOMA-based ALOHA system. The proposed solution ensures Nash equilibrium, but not necessarily the maximum throughput. Multi-channel slotted ALOHA has been studied in~\cite{seo2018nonorthogonal,choi2017noma}, where different users are distributed in different channels, and NOMA is applied when more than one user access the same channel at the same time. In~\cite{choi2017noma}, network throughput is analyzed, assuming that all users achieve the same data rate. \cite{seo2018nonorthogonal} extends the work in ~\cite{choi2017noma} by considering wireless fading channels. For NOMA transmissions in a random access network, users usually do not know other users' channel status and it is difficult to control the transmission power of users to ensure all users achieve the same data rate. In other words, in a random access network, user's data rate is determined by the received signal to noise and interference ratio (SINR), with SIC if applicable. To the best of our knowledge, none of the existing works analyze the performance of a slotted ALOHA system, considering variable rates of users due to the lack of efficient power control. In addition, no existing works study how to adapt protocol parameters to achieve the maximal network performance.

In this paper, we study the performance of a $p$-persistent slotted ALOHA system where NOMA transmissions are supported. Specifically, in each slot, a user decides to transmit data with probability $p$, or do not transmit with probability $1-p$. For users who transmit, they further select a power level, high or low, with certain probabilities. NOMA technique is applied to decode the data when multiple users transmit in the same slot. NOMA transmissions can be successfully decoded only when the involved users do not severely interfere with each other, i.e., SINR, with SIC if applicable, is larger than a threshold. Our objective is to maximize the network throughput of the ALOHA system, by optimizing the transmission probability of each user using different power levels. To this end, we first develop a mathematical model to analyze the performance of the system in terms of the success probability of NOMA transmissions, and the average network throughput. Second, we derive the sufficient and necessary conditions to ensure successful NOMA transmissions. The feasible region of the maximum number of supported high power users and low power users is quantified. Third, based on the analysis, the optimal probability of selecting high/low powers for data transmissions can be obtained to maximize the network throughput. Last but not the least, the impacts of the power levels and SINR threshold on the system performance is further investigated. 



The remainder of the paper is organized as follows. The system model is presented in Section~\ref{sec:systemmodle}. An analytical model is developed in Section~\ref{pack_th}, followed by the numerical results in Section~\ref{numerical results}. Concluding remarks are provided in Section~\ref{conclusion}.

\section{System Model}\label{sec:systemmodle}

We consider a $p$-persistent slotted ALOHA system of $m$ users and an AP. Users are randomly deployed in the coverage area of the AP, i.e., a circle with radius $R$ meters. Define $r_n$ as the communication distance between the $n$-th user and the AP. Correspondingly, the channel gain between the  $n$-th user and the AP is $g_n=L_0 r_n^{-\alpha}|h_n|^2$, where $\alpha, L_0$ represent the path loss exponent and the reference path loss respectively, and $|h_n|^2$ is the small scale fading. 

Users adopt $p$-persistent slotted ALOHA protocol to transmit data to the AP. In each time slot, a user transmits with probability $p$, and do not transmit with probability $1-p$. To facilitate NOMA transmissions, each transmitting user can further choose to adapt its transmission power to ensure the received signal strength at the AP is $v_1$ or $v_2$ ($v_1 > v_2$), with different probabilities. Therefore, in each time slot, a user chooses to adapt its transmission power to $v_1$ with probability $\tau_1$, and chooses to adapt to $v_2$ with probability $\tau_2$, and $\tau_1+\tau_2=p.$ In a random access network, a user can measure the statistics of the channel between itself and the AP, but does not have the channel information of other users. Thus, the $n$-th user can choose its transmission power to ensure the received signal strength at the AP,
\begin{align}
P_{n}\!\!=\!\!
\begin{cases}
v_{1}/g_n,& \mbox{high power mode with prob. $\tau_1$},\\
v_{2}/g_n,& \mbox{low power mode with prob. $\tau_2$},\\
0, & \mbox{no Tx with prob. $1\!-\!\tau_1\!-\!\tau_2$}, 
\end{cases}
\end{align}
where the transmission power is inversely proportional to the channel gain.

Due to random access, it is possible that multiple users may transmit in the same time slot, choosing either $v_1$ or $v_2$. In such case, the AP applies NOMA technique to decode the data. 
Denote the subsets of users selecting $v_1$ and $v_2$ for data transmissions by $\mathcal{N}_h$ and $\mathcal{N}_l$, respectively. The received signal at the AP is written as
\begin{align}
    y=\sum_{i\in\mathcal{N}_h}\sqrt{v_1}x_i+\sum_{j\in\mathcal{N}_l}\sqrt{v_2}x_j+z,
\end{align}
where $x_i$, $x_j$ are the transmitted signals of the user $i$ and user $j$, for $i\in\mathcal{N}_h$ and $j\in\mathcal{N}_l,$ and $z$ is the background noise. 

The AP first decodes the signal with the highest power under the interference from all other users involved in the NOMA transmissions. 
Given that the highest signal is from user set $\mathcal{N}_h$, the SINR of the first decoding user is written as
\begin{align}\label{sinr_req_1}
    \mathbf{SINR}_1=\frac{v_1}{v_1(|\mathcal{N}_h|-1)+v_2(|\mathcal{N}_l|)+1},
\end{align}
where $|\cdot|$ represents the number of users in the corresponding user set. 
If the first $i-1$ users in $\mathcal{N}_h$ are successfully decoded, the received SINR of the $i$-th user in the high power mode is 
\begin{align}\label{sinr_req_h}
    \mathbf{SINR}_i=\frac{v_1}{v_1(|\mathcal{N}_h|-i)+v_2(|\mathcal{N}_l|)+1},
\end{align}
where $i\leq |\mathcal{N}_h|$.
Notice that $\mathbf{SINR}_i$ in (\ref{sinr_req_h}) increases with $i$. This implies that if the first signal is successfully decoded, all the remaining signals of users in high power mode will also be decoded.  

After all signals of users in high power mode are canceled, the AP continues to decode signals of users in low power mode if the user set $\mathcal{N}_l$ is not empty. Similarly, the SINR of the $j$-th users in the low power mode is
\begin{align}\label{sinr_req_l}
    \mathbf{SINR}_j=\frac{v_2}{v_2(|\mathcal{N}_l|-j)+1}, 
\end{align}%
where $j\leq |\mathcal{N}_l|$.
Similarly, $\mathbf{SINR}_j$ in (\ref{sinr_req_l}) increases with $j$, implying that if the first signal in $\mathcal{N}_l$ can be  successfully decoded, all the following signals of users in low power mode will also be decoded. The AP continues to apply SIC to decode signals until one signal cannot be successfully decoded or all signals are successfully decoded.

\section{Performance Analysis}\label{pack_th}

In this section, we develop an analytical model to analyze the performance of $p$-persistent slotted ALOHA where NOMA transmission is supported. We quantify the feasible region of $|\mathcal{N}_h|, |\mathcal{N}_l|$ that ensures successful NOMA transmissions. Based on the analysis, the successful transmission probability and the long-term average system throughput are derived as a function of the power levels $v_1,v_2$, the number of contending users $m$, and the SINR threshold $\Gamma$. We further formulate a throughput maximization problem and propose an algorithm to solve the formulated combinatorial optimization problem. 

\subsection{Feasible region of $|\mathcal{N}_h|$ and $|\mathcal{N}_l|$}
To guarantee successful NOMA transmissions, the received SINR at the AP should be larger than a threshold $\Gamma$. When a number of users concurrently transmit, it is possible that the SINR of some users may not meet the requirement. Thus, we first investigate the feasible region, i.e., the maximum number of users in  $\mathcal{N}_h$ and $\mathcal{N}_l$, wherein the received signals of some or all users can be decoded.


\begin{lemma}\label{ach_reg}
To ensure successful NOMA transmissions, $|\mathcal{N}_h|, |\mathcal{N}_l|$ should satisfy the following conditions:
\begin{itemize}
    \item[a)] The signals of users in high power mode can be decoded when $|\mathcal{N}_h|\leq\min\{m,(\lfloor\frac{v_1-\Gamma}{ v_1\Gamma}\rfloor+1)^\dag\}$, and $|\mathcal{N}_l|\leq \min\{m-{|\mathcal{N}_h|},(\lfloor \frac{v_1-\Gamma (|\mathcal{N}_h|-1) v_1-\Gamma}{v_2\Gamma}\rfloor)^\dag\}$, where $\lfloor x \rfloor$ equals to the biggest integer smaller than or equal to $x$, and $(x)^\dag=max\{0,x\}$. 
    \item[b)] The signals of users in low power mode can be decoded when $|\mathcal{N}_h|\leq \min \{m-|\mathcal{N}_l|,(\lfloor \frac{v_1-|\mathcal{N}_l|v_2\Gamma-\Gamma }{v_1 \Gamma}+1\rfloor)^\dag\}$, and $ |\mathcal{N}_l|\leq \min\{m,(\lfloor\frac{v_2-\Gamma}{v_2\Gamma}\rfloor+1)^\dag\}$.
\end{itemize}

\end{lemma}


\begin{proof}
a) The signal of users in high power mode can be decoded if the signal of the user with the highest signal strength satisfy the SINR condition, i.e., 
\begin{align}\label{v_1_condition}
    SINR_1 = \frac{v_1}{(|\mathcal{N}_h|-1)v_1+|\mathcal{N}_l|v_2+1}\geq \Gamma.
\end{align}%
(\ref{v_1_condition}) is re-written as 
\begin{align} 
    |\mathcal{N}_h|\leq (\lfloor\frac{v_1 - v_2|\mathcal{N}_l|-\Gamma}{ v_1\Gamma}\rfloor+1)^\dag\leq (\lfloor\frac{v_1-\Gamma}{ v_1\Gamma}\rfloor+1)^\dag \nonumber
\end{align}
\begin{align}
|\mathcal{N}_l|\leq (\lfloor \frac{v_1-\Gamma (|\mathcal{N}_h|-1) \nonumber v_1-\Gamma}{v_2\Gamma}\rfloor)^\dag.\nonumber
\end{align}
Given that $|\mathcal{N}_h|+|\mathcal{N}_l|\leq m$, we have $$ |\mathcal{N}_h|\leq\min\{m,(\lfloor\frac{v_1-\Gamma}{ v_1\Gamma}\rfloor+1)^\dag\},$$ and $$|\mathcal{N}_l|\leq \min\{m-{|\mathcal{N}_h|},(\lfloor \frac{v_1-\Gamma (|\mathcal{N}_h|-1) v_1-\Gamma}{v_2\Gamma}\rfloor)^\dag\}.$$ 
Notice that if the first signal can be successfully decoded, all the remaining high power users' signals can be successively decoded by SIC technique, as explain in~\ref{sinr_req_h}. Equivalently, \emph{Lemma}~\ref{ach_reg} a) also ensures that all high power users transmit successfully. 

b) The signals of users in low power mode can be decoded when all signals of users in high power mode are successfully decoded and the signal of the users in the low power node also satisfy the SINR requirement. That is, 
\begin{align}\label{sinr_H}
    \frac{v_1}{(|\mathcal{N}_h|-1)v_1+|\mathcal{N}_l|v_2+1} \geq \Gamma,
\end{align}
\begin{align}\label{sinr_L}
    \frac{v_2}{(|\mathcal{N}_l|-1)v_2+1}\geq \Gamma.
\end{align}%
\eqref{sinr_H} is re-written as 
\begin{align}
|\mathcal{N}_h| \leq (\lfloor\frac{v_1-|\mathcal{N}_l|v_2\Gamma-\Gamma }{v_1 \Gamma}+1\rfloor)^\dag. \nonumber
\end{align}
\eqref{sinr_L} is re-written as
\begin{align}
    |\mathcal{N}_l| \leq (\lfloor\frac{v_2-\Gamma}{ v_2\Gamma}+1\rfloor)^\dag. \nonumber
\end{align} 
Given that $|\mathcal{N}_h|+|\mathcal{N}_l|\leq m,$ we have $$|\mathcal{N}_h|\leq \min \{m-|\mathcal{N}_l|,(\lfloor \frac{v_1-|\mathcal{N}_l|v_2\Gamma-\Gamma }{v_1 \Gamma}+1\rfloor)^\dag\},$$ and $$0\leq |\mathcal{N}_l|\leq \min\{m,(\lfloor\frac{v_2-\Gamma}{v_2\Gamma}\rfloor+1)^\dag\}.$$
Similarly, \emph{Lemma}~\ref{ach_reg} b) also ensures that all low power users transmit successfully. 
\end{proof}

To simplify the notation, we define $n_h=|\mathcal{N}_h|$ and $n_l=|\mathcal{N}_l|.$ To ensure NOMA transmission is successful, at least one user should transmit successfully, either a high power user or low power user. \emph{Lemma}~\ref{ach_reg} a) presents the upper bound of $n_1$ and $n_2$ to ensure that at least one high power user transmits successfully. Define the upper bound of $n_h$ and $n_l$ in \emph{Lemma}~\ref{ach_reg} a) as 
$$n_h^{(a)}=\min\{m,(\lfloor\frac{v_1-\Gamma}{ v_1\Gamma}\rfloor+1)^\dag\},$$ $$n_l^{(a)}=\min\{m-{n_h},(\lfloor \frac{v_1-\Gamma (n_h-1)v_1-\Gamma}{v_2\Gamma}\rfloor)^\dag\}.$$ 
Similarly, Lemma~\ref{ach_reg} b) presents the upper bound of $n_h$ and $n_l$ to ensure that at least one low power user transmits successfully, which is denoted as  
$$n_h^{(b)}=\min\{m-n_l,(\lfloor\frac{v_1-n_lv_2-\Gamma}{v_1\Gamma}\rfloor)^\dag+1\},$$
$$n_l^{(b)}=\min\{m,(\lfloor\frac{v_2-\Gamma}{ v_2\Gamma}\rfloor+1)^\dag\}.$$

\subsection{Successful Transmission Probability}

Given that there are $m$ users using $p$-persistent slotted ALOHA to contend for data transmissions, the probability that $n_1$ users are selecting high power and $n_2$ users are selecting low power to transmit is 
\begin{align}
    &\mathbb{P}(n_1, n_2|m,\tau_1,\tau_2) \nonumber\\
    =& \binom m {n_1+n_2} \binom {n_1+n_2} {n_1} \tau_1^{n_1} \tau_2^{n_2} (1-\tau_1-\tau_2)^{(m-n_1-n_2)}.
\end{align}

To ensure successful NOMA transmissions,  $n_1$ and $n_2$ should satisfy Lemma \ref{ach_reg}. 
\begin{lemma}\label{lm_p_success}
Given a user selects $v_1$ with probability $\tau_1$, $v_2$ with probability $\tau_2$, the probability that its transmission is successful is given by 
\begin{align}\label{p_success}
   \mathbb{P}_{success} &(m,\tau_1,\tau_2) =  \tau_1\sum_{n_1=1}^{n_h^{(a)}}\sum_{n_2=0}^{n_l^{(a)}}\mathbb{P}(n_1-1,n_2|m-1,\tau_1,\tau_2) \nonumber \\
    + & \tau_2\sum_{n_2=1}^{n_l^{(b)}}\sum_{n_1=0}^{n_h^{(b)}}\mathbb{P}(n_1,n_2-1|m-1,\tau_1,\tau_2).
\end{align}%
\end{lemma}

\begin{proof}

If a typical user selects $v_1$,  all high power user's transmissions are successful if \emph{Lemma}~\ref{ach_reg} a) is satisfied, i.e., $n_1\leq n_h^{(a)}$ and $n_2\leq n_l^{(a)}$. Thus, given that the typical user selects $v_1$, the successful transmission probability is
\begin{align}
   \sum_{n_1=1}^{n_h^{(a)}}\sum_{n_2=0}^{n_l^{(a)}}\mathbb{P}(n_1-1,n_2|m-1,\tau_1,\tau_2).
\end{align}%

Similarly,  conditioning on that the typical user selects $v_2$, the conditional successful transmission probability is
\begin{align}
   \sum_{n_2=1}^{n_l^{(b)}}\sum_{n_1=0}^{n_h^{(b)}}\mathbb{P}(n_1,n_2-1|m-1,\tau_1,\tau_2).
\end{align}%
Given the probability that the typical user select $v_1$ and $v_2$ is $\tau_1$ and $\tau_2$ respectively, according to total probability theorem, we can derive the successful transmission probability in~\eqref{p_success}.

\end{proof}

\subsection{Long-term Average Throughput}
Given that $n_1$ users in high power mode and $n_2$ users in low power mode are transmit currently, the long-term average sum rate of all users in the high power mode is
\begin{align}\label{throughput_h}
    \mathbb{E}(R_1|n_1,n_2) = \sum_{i=1}^{n_1}\log_2\left(1+\frac{v_1}{(i-1) v_1+n_2 v_2+1}\right),
\end{align}
where $n_1\leq n_h^{(a)} $ and $n_2 \leq n_l^{(a)}$ according to~\emph{Lemma}~\ref{ach_reg} a).
Similarly, the long-term average sum data rate of all users in low power mode is given by
\begin{align}\label{throughput_l}
    \mathbb{E}(R_2|n_1,n_2) = \sum_{j=1}^{n_2} \log_2\left(1+\frac{v_2}{(j-1)v_2+1}\right),
\end{align}
where $n_1\leq n_h^{(b)} $ and $n_2 \leq n_l^{(b)}$ according to~\emph{Lemma}~\ref{ach_reg} b).

\begin{lemma}\label{throughput_avg}
The long-term average throughput of the system is given by
\begin{align}\label{throughput}
    Th_{avg}&(m,\tau_1,\tau_2) = \sum_{n_1=1}^{n_h^{(a)}}\sum_{n_2=0}^{n_l^{(a)}}\mathbb{E}(R_1|n_1,n_2)\mathbb{P}(n_1,n_2|m,\tau_1,\tau_2) \nonumber\\ 
        + &\sum_{n_2=1}^{n_l^{(b)}}\sum_{n_1=0}^{n_h^{(b)}}\mathbb{E}(R_2|n_1,n_2)\mathbb{P}(n_1,n_2|m,\tau_1,\tau_2),
\end{align}
\end{lemma}

The logic of the proof of \emph{Lemma}~\ref{throughput_avg} is similar to that of the \emph{Lemma}~\ref{lm_p_success}.
 



\subsection{Throughput Maximization Problem}
Our objective is to maximize the achieved long-term average throughput of the ALOHA system by tuning the transmission probability $\tau_1$, $\tau_2$. 
\begin{equation} \label{eq:average_th}
\begin{aligned}
& \underset{\tau_1,\tau_2}{\text{maximize}}
& &Th_{avg}(m,\tau_1,\tau_2) \\
& \text{subject to}
& & 0 < \tau_1,\tau_2 < 1 \\
\end{aligned}
\end{equation}


The formulated optimization problem in~\eqref{eq:average_th} is a combinatorial optimization problem,  which is known to be NP-hard. Thus, we propose an iterative algorithm to obtain the solution in Algorithm~\ref{alg}. We first start from $\tau_1=0$, then we update $\tau_2$ which maximizes the throughput in~\eqref{throughput}; with the updated $\tau_2$, we update $\tau_1$ that maximizes~\eqref{throughput}. We iteratively update $\tau_1$ and $\tau_2$ until the throughput improvement becomes negligible. 

\begin{algorithm}\label{alg}
\SetAlgoLined
\KwResult{$\tau_1$, $\tau_2$, $Th_{avg}$ in \eqref{eq:average_th}}
 Input: $\Gamma$, $m$, $v_1$, $v_2$,\;
 Initialize: $\tau_1=0$, $\epsilon = 10^{-5}$, $Th_{avg}=0$\;
 \While{True}{
  
  For given $\tau_1,$ find $\tau_2$ that maximizes  $Th_{avg}^\prime$=\eqref{throughput}\;   
  \eIf{$Th_{avg}^\prime -Th_{avg} > \epsilon$}{
   Update $\tau_2$;
   $Th_{avg} = Th_{avg}^\prime$\;
   }{
   Break\;
  }
  For updated $\tau_2,$  find $\tau_1$ that maximizes  $Th_{avg}^\prime$=\eqref{throughput}; 
   \eIf{$Th_{avg}^\prime -Th_{avg} > \epsilon$}{
   Update $\tau_1$;
   $Th_{avg} = Th_{avg}^\prime$\;
   }{
   Break\;
  }
   }
 \caption{Algorithm to find the solution for \eqref{eq:average_th}}
\end{algorithm}


\section{Numerical results}\label{numerical results}
In this section, simulations are conducted with Matlab to validate the analysis. We setup a network and implement $p$-persistent slotted ALOHA,  as described in Sec.~\ref{sec:systemmodle}. If not otherwise specified, we set $v_1 = 4$, $v_2 = 1.5$, $m=10$, $\Gamma = 1.5$. 


\begin{figure*}[!htb]
  \centering
  \hfill
 \begin{minipage}{0.27\textwidth}
    \centering
\includegraphics[width=\textwidth]{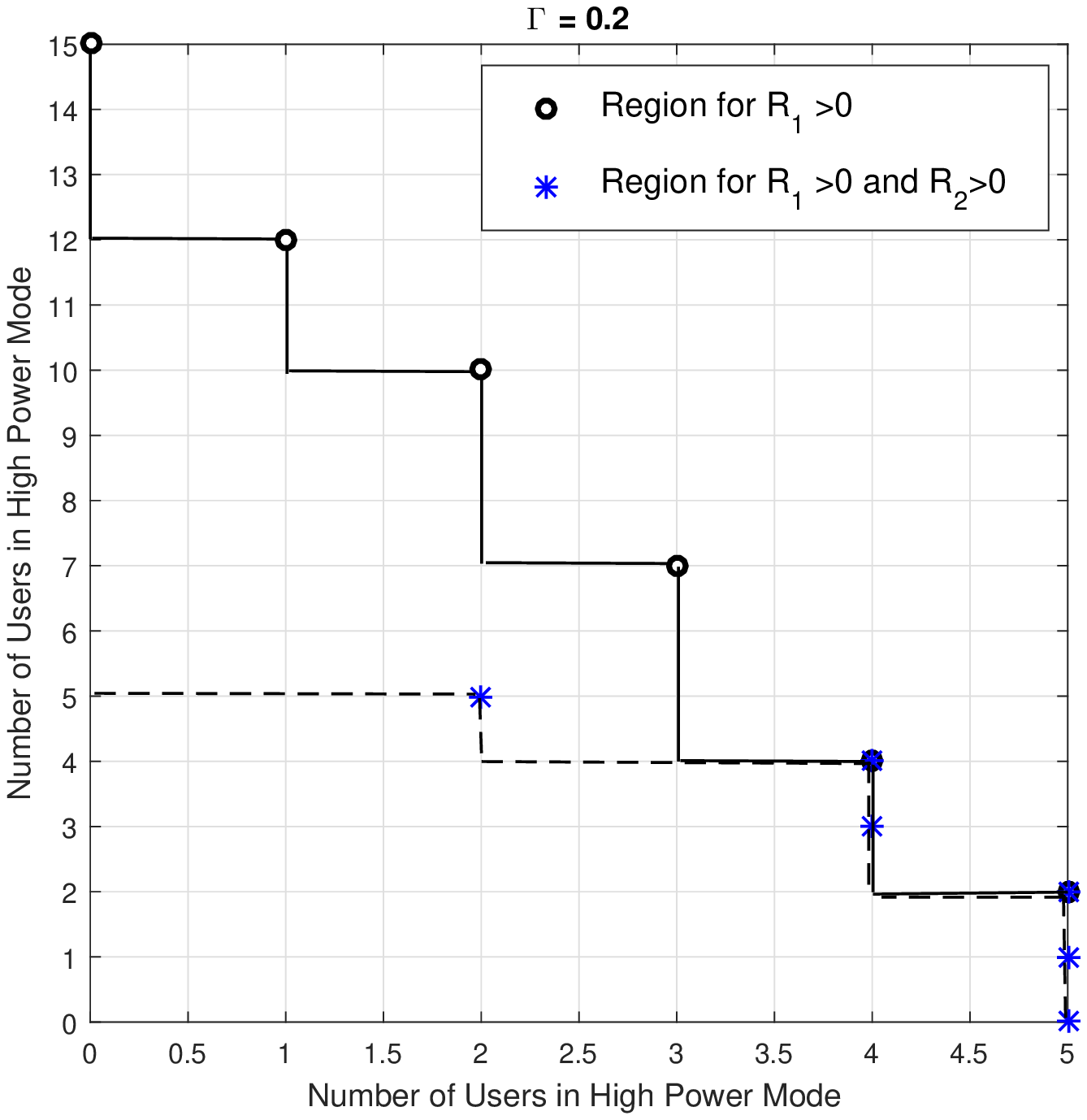}
 \caption{Feasible Regions in Lemma 1}
 \label{Feasible_region}
  \end{minipage}
  \hfill
 \begin{minipage}{0.27\textwidth}
  	\centering
    \includegraphics[width=\linewidth]{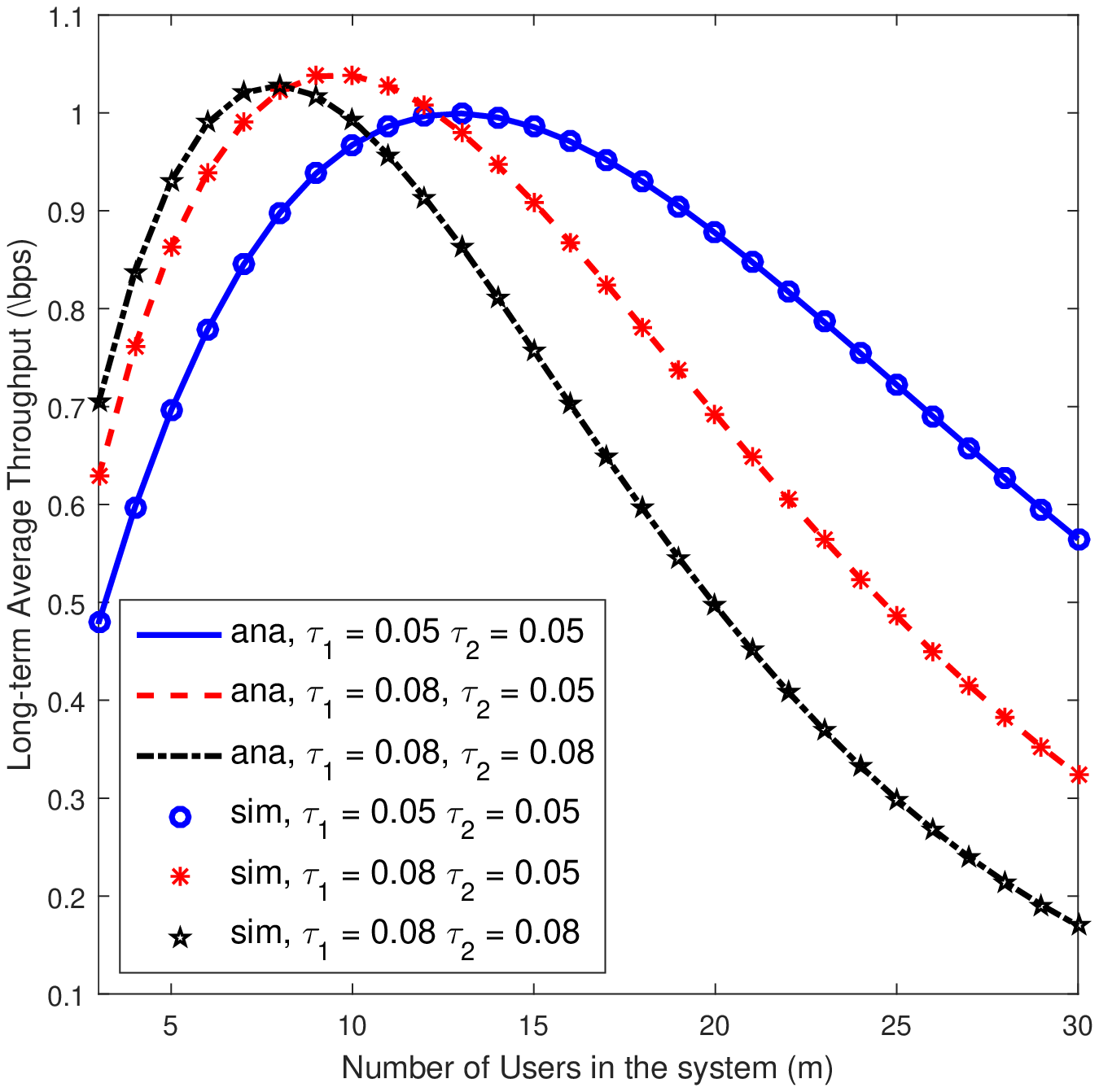}
    \caption{Long-term Average Throughput under different number of users  $m$.}
\label{relationship_m_th}
  \end{minipage}
  \hfill
\begin{minipage}{0.27\textwidth}
    \centering
    \includegraphics[width=\textwidth]{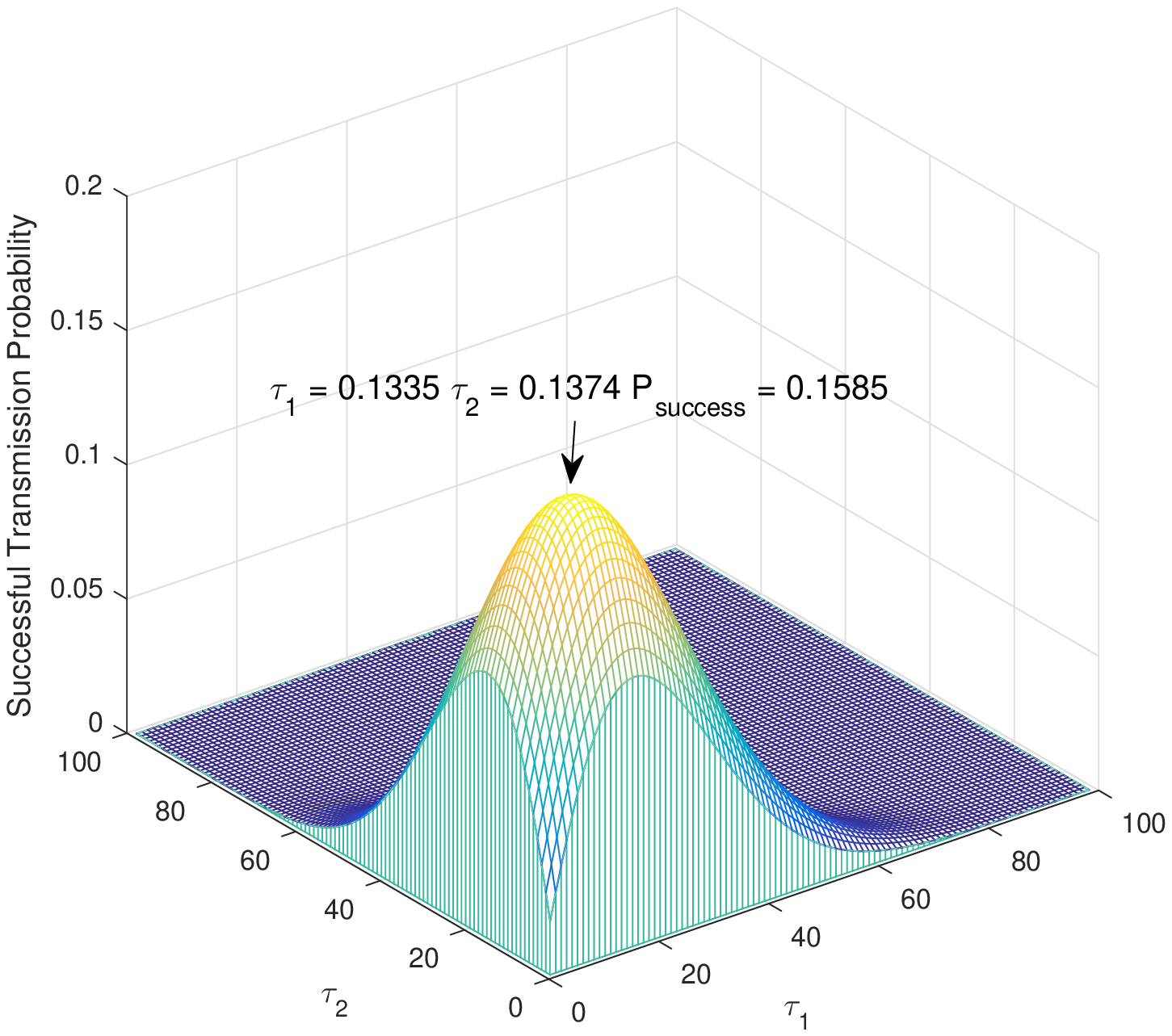}
\caption{Successfully Transmission probability under different $\tau_1$ and $\tau_2$.}
\label{p_s_t_2}
  \end{minipage}
  \hfill
\end{figure*}

\begin{figure*}[!htb]
  \centering
  \hfill
   \begin{minipage}{0.27\textwidth}
  	\centering
\includegraphics[width=\textwidth]{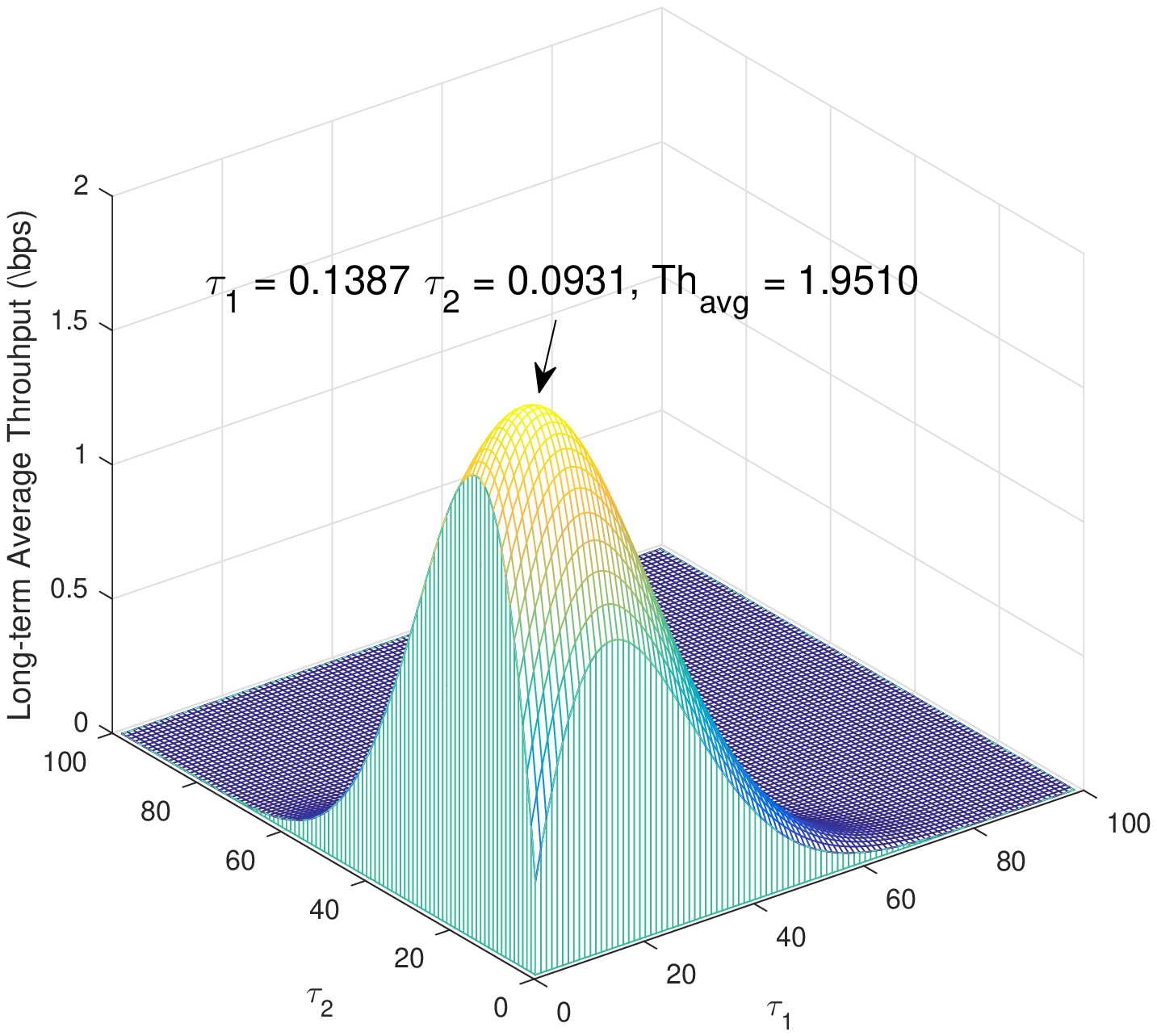}
\caption{Long-term average throughput under different $\tau_1$ and $\tau_2$.}
\label{average_th_1}
  \end{minipage}
  \hfill
  \begin{minipage}{0.27\textwidth}
    \centering
    \includegraphics[width=\linewidth]{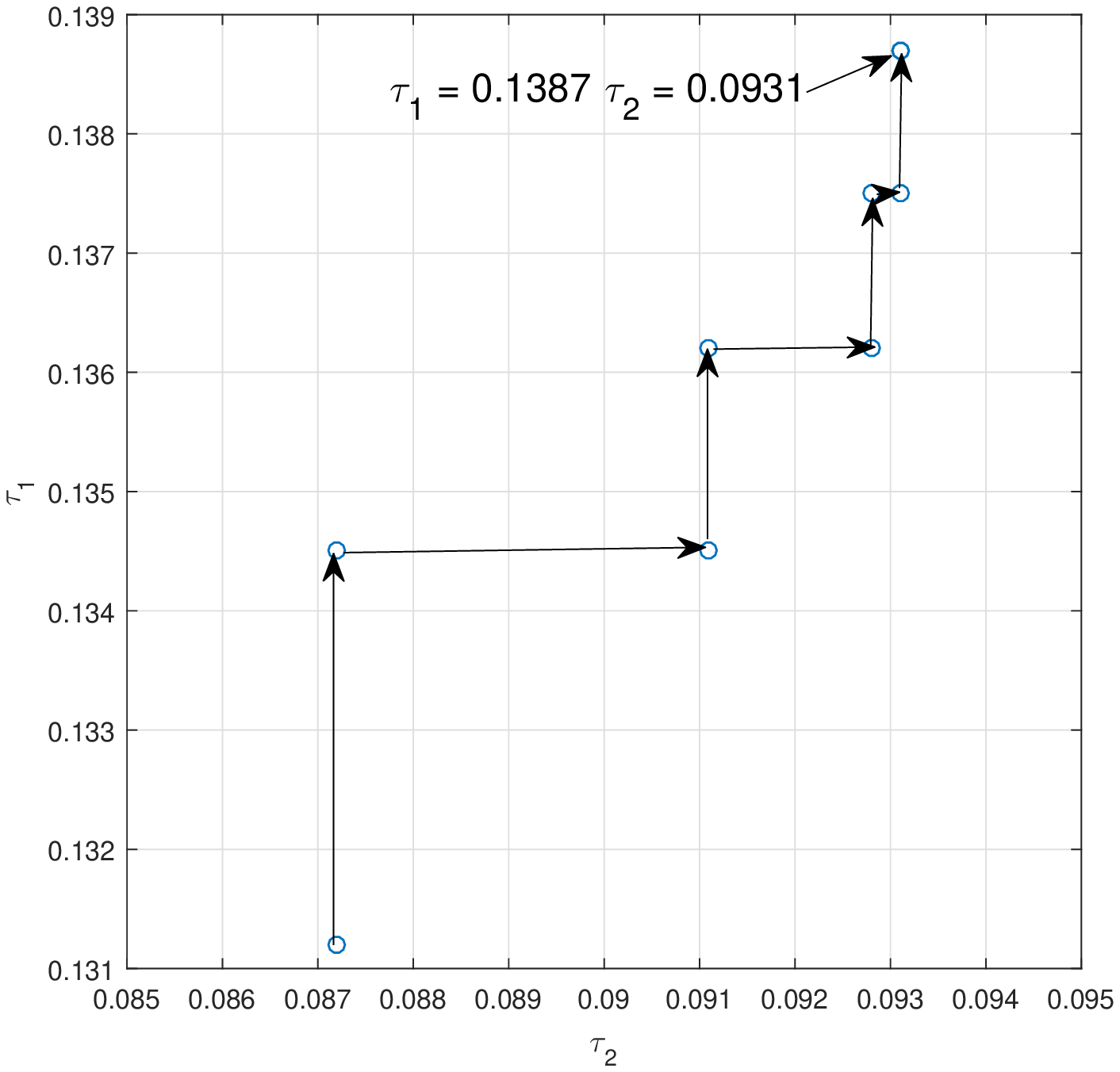}
    \caption{Optimal solution with Algorithm 1}
\label{fig:alg}
  \end{minipage}
  \hfill
 \begin{minipage}{0.27\textwidth}
    \centering
    \includegraphics[width=\linewidth]{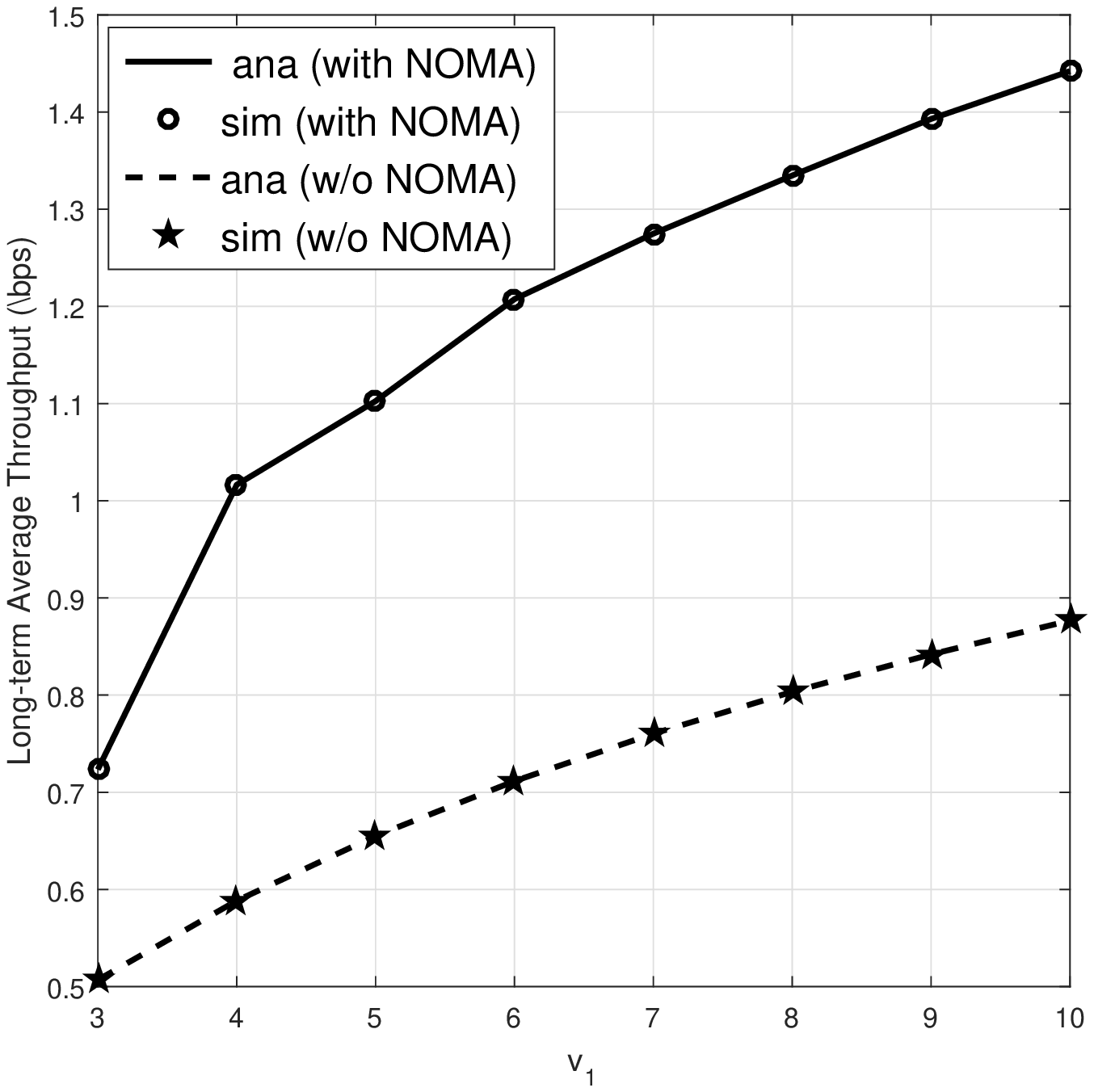}      
      \caption{Average throughput with and w/o NOMA}
\label{fig:compare_with_class_th}
  \end{minipage}
  \hfill
\end{figure*}

The feasible region presented in Lemma 1 in shown in Fig.~\ref{Feasible_region}. The feasible region to ensure high power users transmission is illustrated in solid line while the feasible region to ensure low power users transmission is in dotted line. When there are $n_1=2$ high power users and $n_2=5$ low power users, both high power users and low power users' data can be decoded, and thus $R_1>0$ for high power user and $R_2>0$ for low power user. If $5 < n_2 \leq 10 $, high power users data still can be decoded but the data rates decrease; low power users data cannot be decoded due to severe interference from concurrent transmissions and the SINR requirement cannot be satisfied. 

The long term average throughput under different number of users is shown in Fig.~\ref{relationship_m_th}. It can be seen that the throughput first increases with $m$ as more users achieves a higher multiplexing gain; and then decreases when more users joins the network, as more users increases the contentions that degrade the channel utilization efficiency. The peak point shifted to right for a smaller $\tau_1$ and $\tau_2$, when users has a smaller transmission probability. The simulations validate the analysis.  

Fig.~\ref{p_s_t_2} and Fig.~\ref{average_th_1} plot the successful transmission probability and throughput under different  $\tau_1$ and $\tau_2$ using exhaustive search. It can be seen that the optimal settings of $\tau_1$ and $\tau_2$ to achieve the maximum transmission probability and the maximum throughput are different. \eqref{eq:average_th} jointly considers the successful transmission probability in \eqref{p_success} and the achieved data rates of NOMA transmissions. It is also observed that  $\tau_1 > \tau_2$ when the maximum throughput is achieved, as it is more desirable for users to choose high power mode to achieve a greater throughput. 
The same solution of optimal $\tau_1$ and $\tau_2$ can be found in 10 iterations using Algorithm~\ref{alg}. 

The throughput performance of slotted ALOHA with and w/o NOMA is compared in Fig.~\ref{fig:compare_with_class_th}, under the optimal setting of $\tau_1$ and $\tau_2$. For conventional $P$-persistent ALOHA w/o NOMA, the successful transmission probability is
$   \mathbb{P}_{success}(m,p) =  p(1-p)^{m-1},$
and the long-term average throughput is $Th_{avg}(m,p) = \log_2(1+v_1)p(1-p)^{m-1}$. The maximum throughput is achieved when $p=\frac{1}{m}$. It can be seen that the throughput of ALOHA with NOMA significantly outperforms that of conventional ALOHA w/o NOMA. 

\vspace{-8pt}
\section{Conclusion}\label{conclusion} \vspace{-3pt}
In this paper, we have analyzed the performance of NOMA in a random access network using p-persistent slotted ALOHA. To ensure successful NOMA transmissions, the feasible region of the maximum number of users using high and low power has been derived. To achieve the maximal long-term average throughput, an algorithm is further proposed to find the optimal probabilities using high and low power for data transmissions. The impacts of the transmission powers on long-term average system throughput has been studied as well.

In this work, we optimize the transmission parameters to maximize the system performance, given that the wireless network environment is known by each individual user. In our future work, we will apply learning techniques to allow users to distributively learn the unknown wireless environment to make the optimal decision to maximize the throughput. 

\vspace{-6pt}
\bibliographystyle{IEEEtran}
\bibliography{IEEEfull,geo}
\vspace{-15pt}
\end{document}